\documentclass{article}

\usepackage[nonatbib,preprint]{neurips_2021}

\usepackage[utf8]{inputenc} % allow utf-8 input
\usepackage[T1]{fontenc}    % use 8-bit T1 fonts
\usepackage{hyperref}       % hyperlinks
\usepackage{url}            % simple URL typesetting
\usepackage{booktabs}       % professional-quality tables
\usepackage{amsfonts}       % blackboard math symbols
\usepackage{nicefrac}       % compact symbols for 1/2, etc.
\usepackage{microtype}      % microtypography
\usepackage{xcolor}         % colors
\usepackage[numbers]{natbib}
\usepackage{graphicx}
\usepackage{mathtools}
\usepackage{amsmath}
\usepackage{tikz}
\usetikzlibrary{calc}
\usepackage{enumitem}
\usepackage{amsthm}
\usepackage{thm-restate}
\usepackage[capitalise,noabbrev]{cleveref}
\usepackage{booktabs}
\usepackage{multirow}
\usepackage{newfloat}
\usepackage{wrapfig}
\usepackage{dsfont}
\usepackage{comment}
\usepackage{multirow}
\usepackage{mleftright}
\usepackage{mdframed}
\usepackage{bm}
\usepackage{multirow}
\usepackage{derivative}
\usepackage[mathscr]{eucal}
\usepackage[ruled]{algorithm2e} 

\usepackage{titlesec}
\titlespacing*{\paragraph} {0pt}{0.35ex plus 0.3ex minus .2ex}{1em}

\newtheorem{theorem}{Theorem}

\newtheorem{definition}{Definition}
\newtheorem{assumption}{Assumption}

\DeclareMathOperator{\co}{co}
\newcommand{\distr}{\mathscr{P}}
\newcommand{\X}{\mathcal{X}}

\newcommand{\I}{\mathcal{I}}

\newcommand{\expe}{\mathbb{E}}
\newcommand{\bb}[1]{\bm{#1}}

\newcommand{\dom}{\mathcal{D}}

\newcommand{\xvec}{\bb{x}}
\newcommand{\yvec}{\bb{y}}

\newcommand{\cvec}{\bb{c}}
\newcommand{\subvec}{\bb{\sub}}
\newcommand{\sub}{g}
\newcommand{\lambdavec}{\bb{\lambda}}
\newcommand\norm[1]{\mleft\lVert #1 \mright\rVert}

\newcommand{\opt}{\textsc{Opt}}
\newcommand{\rew}{\textsc{Rew}}
\newcommand{\reg}{\textsc{Reg}}
\DeclareMathOperator*{\argmax}{arg\,max}
\DeclareMathOperator*{\argmin}{arg\,min}

\newcommand{\defeq}{\mathrel{:\mkern-0.25mu=}}

\title{The Parity Ray Regularizer for Pacing in Auction Markets}

\author{%
  Andrea Celli \\
  Facebook Core Data Science\\
  \texttt{andrea.celli@polimi.it} \\
\And
  Riccardo Colini-Baldeschi\\
  Facebook Core Data Science\\
  \texttt{rickuz@fb.com} \\
\And
  Christian Kroer\\
  Columbia University\\
  \texttt{christian.kroer@columbia.edu} \\
\And
  Eric Sodomka\\
  Facebook Core Data Science\\
  \texttt{sodomka@fb.com} \\
}

\begin{document}

\maketitle

\begin{abstract}
\emph{Budget-management systems} are one of the key components of modern auction markets. Internet advertising platforms typically offer advertisers the possibility to pace the rate at which their budget is depleted, through \emph{budget-pacing mechanisms}. 
We focus on \emph{multiplicative pacing} mechanisms in
an online setting in which a bidder is repeatedly confronted with a series of advertising opportunities. After collecting bids, each item is then allocated through a single-item, second-price auction. If there were no budgetary constraints, bidding truthfully would be an optimal choice for the advertiser. 
However, since their budget is limited, the advertiser may want to shade their bid downwards in order to preserve their budget for future opportunities, and to spread expenditures evenly over time.
The literature on online pacing problems mostly focuses on the setting in which the bidder optimizes an additive separable objective, such as the total click-through rate or the revenue of the allocation. In many settings, however, bidders may also care about other objectives which oftentimes are non-separable, and therefore not amenable to traditional online learning techniques. 
Building on recent work, we study the frequent case in which advertisers seek to reach a certain distribution of impressions over a target population of users.
We introduce a novel regularizer to achieve this desideratum, and show how to integrate it into an online mirror descent scheme attaining the optimal order of sub-linear regret compared to the optimal allocation in hindsight when inputs are drawn independently, from an unknown distribution. 
Moreover, we show that our approach can easily be incorporated in \emph{standard} existing pacing systems that are not usually built for this objective. 
The effectiveness of our algorithm in internet advertising applications is confirmed by numerical experiments on real-world data.

\end{abstract}

\section{Introduction}\label{intro}

In the last decade, the spending on Internet advertising has grown dramatically, reaching more than \$130 billion in the United States in 2019~\cite{spending}. One of the drivers of this rapid growth has been the success of auction markets as a practical mechanism to match bidders (i.e., advertisers) to their target audience at an appropriate price~\cite{edelman2007internet,varian2007position}. 
In this type of mechanisms, advertisers usually specify a targeting rule to target a certain population of users, and they are asked to bid on events of interest such as an impression, a click, a conversion or a video view. 
For each ad slot, the mechanism determines a winning advertiser, which is subsequently given the chance to show an impression and potentially generate the event of interest. 

Auction markets are typically operationalized through simple mechanisms that can be effectively implemented at scale, and that simplify the interaction between advertisers and the platform. 
A common paradigm to manage large amounts of repeated ad auctions is for the platform to operate a \emph{proxy bidder} on behalf of each advertiser. The advertisers communicate the maximum bid for the event of interest, their targeting criteria, and their overall budget for their campaigns to the proxy bidder. 
In real time, the proxy bidder then constructs bids for each individual auction via these parameters, as well as additional information computed by the platform.

The complexity of the competitive interactions taking place between advertisers, the frequency of the decisions, and the intrinsic uncertainties of the environment make budget management particularly challenging in this setting.
An advertiser may lose out on significant amounts of revenue if their budget is depleted too early (thereby missing potentially valuable future bidding opportunities), or if their budget it not fully depleted within the planned duration of the campaign. 
To solve this problem, platforms offer advertisers the possibility to {\em pace} the rate at which their budget is depleted, through a {\em budget-pacing mechanism}. 
The most common pacing mechanisms modify the values of the bids within a series of repeated auctions either by shading the bid appropriately (i.e., \emph{multiplicative pacing})~\cite{balseiro2017budget,conitzer2021multiplicative}, or by determining a participation probability through the toss of an appropriately weighted coin for each auction (i.e., \emph{probabilistic pacing})~\cite{agarwal2014budget,balseiro2017budget}. 
In this work, we focus on the former type of budget-management mechanism. This choice is motivated by the widespread adoption of this mechanisms in large online advertising platforms.
Specifically, we study an online setting in which an advertiser is repeatedly confronted with a series of advertising opportunities, one per iteration, allocated through a second-price auction. 
If there were no budgetary constraints, bidding truthfully (i.e., bidding their true valuation for an advertising opportunity) would be an optimal choice for the advertiser. However, since their budget is limited, the advertiser may want to shade their bid downwards in order to preserve their budget for future opportunities. 

The literature on online bidding and online pacing problems mostly focuses on settings in which the bidder optimizes an \emph{additively separable objective}, such as the total click-through rate or revenue of the allocation (see, e.g.,~\cite{balseiro2020best,balseiro2019learning,borgs2007dynamics,conitzer2021multiplicative,feldman2007budget,hosanagar2008optimal}). In many settings, however, advertisers may also care about other objectives which oftentimes are \emph{non-separable}, and therefore not amenable to traditional online learning techniques. 
Two notable exceptions are the work by~\citet{agrawal2014fast}, which focuses on solving general online stochastic convex programs with general concave objectives and convex constraints, and the recent work by~\citet{balseiro2020regularized}, where the authors study online allocation problems with non-separable objectives in settings such as fairness across advertisers and load balancing.\footnote{We devote~\cref{sec:related works} to a more in depth discussion of relevant previous works.}
We argue that another frequent objective, which is oftentimes sought after by advertisers, is reaching a certain distribution of impressions over a target population of users. 
Surprisingly enough, to the best of our knowledge, industry practitioners don’t have at their disposal a practical way to solve this problem.
Consider, for example, the following real-world use cases: an advertiser may want to enforce uniform sampling on a target population for conducting unbiased surveys, 
to promote a business based on a two-sided market by balancing their reach on users belonging to ``both sides of the market'' (this is the case, for example, of dating apps, and online marketplaces), or to perform online outreach to a population of users.
In the latter setting,~\citet{gelauff2020advertising} showed that advertisers have to resort to complex segmentation strategies through subcampaigns to achieve a distribution of conversions close to that of the overall adult population of a target city.

\paragraph{Our contributions} Thus motivated, we present a practical way to steer a multiplicative-pacing mechanism so as to satisfy advertisers' distributional preferences over the advertising opportunities which they are allocated. In particular, an advertiser may want to reach a realized distribution of impressions close to some target distribution across a given user breakdown.
We refer to the feature determining the breakdown of interest as the \emph{category} of the user. 
We consider a finite horizon model in which an advertiser participates in a series of auctions. The advertiser has a fixed amount of non-replenishable resources (i.e., a budget). We consider an incomplete information model where, at each iteration, an item is drawn i.i.d. from some unknown distribution and presented to the advertiser, who has to compute a bid for the current item. Then, a second-price auction is run and the advertiser observes their reward and budget consumption. The advertiser does not get to observe the reward functions and types of future requests until their arrival. 
First, we describe a regularizer, which we call the \emph{parity ray regularizer}, that can be employed to model the advertiser's distributional preferences. 
Second, we propose a practical mechanism for our specific regularized online pacing problem by casting the online mirror descent scheme  by~\citet{balseiro2020regularized} to our setting.
Our algorithm guarantees regret of order $O(T^{1/2})$, which optimal in our stochastic setting.
Moreover, we show that our framework could easily be incorporated into standard existing pacing systems that are not usually built for this objective.
Finally, we empirically demonstrate the efficacy of our mechanism on real-world data from a large Internet advertising company.

\section{Multiplicative pacing for auction markets}
\label{sec:pacing in auctions}

We study the problem faced by a bidder (i.e., an advertiser) who wishes to maximize their utility over a sequence of auctions that occur over $T$ time steps.
At each time $t=1,\ldots,T$ an item (i.e., an advertising opportunity to reach a user) appears, and an auction is run in order to allocate it. We will assume that the item is allocated via a single-item \emph{second-price auction}, meaning that the highest bid wins, and the winner pays the second-highest bid.
At every time $t$, the bidder observes their valuation $v_t$ for the item, submits a bid $b_t$, and then observes the variable $x_t\in \{0,1\}$ specifying whether they won, and if they did win then they observe the price $p_t$ that they paid for the item. 

Without further restrictions, bidding in such a dynamic auction setting would be straightforward: the bidder can simply bid their true value, and the second-price auction ensures that this maximizes their utility.
However, in practice, bidders usually have a total budget $B$, and they wish to constrain their expenditure across the $T$ auctions to be at most $B$. This complicates bidding even in second-price auctions, because this budget constraint links all $T$ auctions together.\footnote{Here, we assume that the bidder cannot manipulate future prices by bidding higher when they lose, in order to encourage the depletion of the budget of other bidders. This assumption is reasonable in large-scale markets, as we will discuss in~\cref{sec:problem description}.}
If we allow fractional allocations, then an omniscient budget-constrained bidder can be thought of as a fractional knapsack packer: they now wish to order the $T$ auctions by \emph{decreasing bang-per-buck} $v_t / p_t$, and grab items in this order until they either deplete their budget, or they reach some item $t$ such that $v_t / p_t \leq 1$~\cite{conitzer2021multiplicative}. 
In practice, the bidder must decide their bids in an online fashion: at time $t$ they only know their first $t$ valuations $(v_1,\ldots,v_t)$, their past allocations up to time $t-1$, i.e., $(x_1,\ldots x_{t-1})$, as well as their past expenditures $(x_1 p_1,\ldots,x_{t-1}p_{t-1})$.
Thus, the only information they have about the auction at time $t$ is their valuation for the item, and they must adaptively smooth out their spending across the $T$ auctions, in order to satisfy their budget constraint while maximizing utility.

The above problem of how to bid in such a setting has been studied extensively both from the perspective of an individual bidder, as well as from a platform perspective. One popular approach in practice is to use a \emph{pacing multiplier} $\alpha \in [0,1]$, which is then used to construct each bid as $b_t = \alpha v_t$. Given a sequence of prices $(p_1,\ldots, p_T)$, an optimal allocation can be achieved within the context of second-price auctions by choosing $\alpha$ such that the bidder spends their budget exactly (assuming that the bidder can choose the fraction that they wish to win on any items such that $\alpha v_t = p_t$)~\citep{balseiro2015repeated,conitzer2021multiplicative}. The pacing multiplier can be adaptively controlled via online-learning methods in order to guarantee asymptotic optimality in stochastic environments~\citep{balseiro2019learning,balseiro2020best}.

In the above, we presented the setting as if the individual bidder controls their pacing parameter $\alpha$ used in smoothing their expenditure. However, a second model which arises in practice is the \emph{proxy-bidder setting}. The proxy-bidder setting arises in online advertising, where advertisers often submit only a value per click, budget, and targeting criteria to the platform. The platform then performs budget smoothing on behalf of the advertiser. In that setting, the proxy bidder refers to a bidder operated by the platform, which is effectively a control algorithm that scales the parameter $\alpha$ up or down, depending on whether the advertiser is on track to spend their budget correctly across the given time frame.
Equilibria arising in the proxy-bidder setting have been studied in several works~(see, e.g., \citep{balseiro2017budget,conitzer2021multiplicative,conitzer2019pacing}).

\section{Problem description}\label{sec:problem description}

In this paper, we focus on the setting where we wish to smooth out the expenditure of a single bidder, while the rest of the market behaves according to a stochastic model. In particular, we will assume that valuations, prices, and user categories are generated i.i.d. from an unknown distribution.
In practice, valuations and categories can be safely described as coming from a stochastic model. The prices, however, typically come from the paced bids arising from other bidders whom are also employing some budget-smoothing technique. In small-scale settings where bidders react to the expenditure of other bidders, these prices could behave adversarially rather than stochastically. However, for large-scale markets, an individual bidder has almost no impact on the prices or the overall market, in which case stochastic behavior of the prices is a reasonable assumption. This idea can be formalized, for example, via fluid mean-field models such as the one given by \citet{balseiro2015repeated}.

The bidder can spend a budget $B$ over a sequence of auctions withing a finite time horizon $T$. Let $\xvec\in[0,1]^T$ be the allocation vector for the bidder, and $x_t\in\X_t\subseteq[0,1]$ be the allocation at time $t$. Let $\rho\in\mathbb{R}_{>0}$ be the \emph{per-iteration budget}, which is a constant such that $B=T\rho$.
Items can be of $m$ different \emph{types} or, equivalently, \emph{categories}. 
The vector $\cvec_t\in[0,1]^m$ specifies the categories of the item being auctioned at time $t$. At each iteration $t$, the bidder receives as input the tuple $(v_t,p_t,\cvec_t)$ specifying the valuation for the item $v_t\in\mathbb{R}_{\ge0}$, its price $p_t\in\mathbb{R}_{\ge0}$, and its category vector $\cvec_t$. At each $t$, the input tuple is generated i.i.d. from an unknown distribution $\distr\in\Delta^{\I}$, where $\I$ is the set of all possible input tuples that can be observed.\footnote{
The set $\{1,\ldots,n\}$, where $n\in\mathbb{N}_{\ge0}$, is compactly denoted as $[n]$; the empty set as $\emptyset$. Given a set $\X$, we denote its convex hull with the symbol $\co {\X}$. Vectors are marked in bold. We denote by $\langle\cdot,\cdot\rangle$ the scalar product between two vectors.
Given a discrete set $\X$, we denote by $\Delta^{\X}$ the $|\X|$-simplex, that is, the set $\Delta^{\X} \defeq \{\bb{v}\in\mathbb{R}_{\ge0}^{|\X|}: \sum_{i} v_i = 1\}$. Analogously, we denote the \emph{full-dimensional standard simplex} over $\X$ by $\Delta_{+}^{\X}\defeq\{\bb{v}\in\mathbb{R}_{\ge0}^{|\X|}: \sum_{i} v_i \le1\}$. The symbols $\Delta^n$ and $\Delta^n_+$, with $n\in\mathbb{N}_{>0}$, are used to denote $\Delta^{[n]}$ and $\Delta^{[n]}_+$, respectively.
}
In the remainder of the paper we make the following assumption on the regularity of inputs, which is easily satisfied in practice.
\begin{assumption}\label{assumption:upper bounds}
There exists $\bar v\in\mathbb{R}_{\ge0}$ and $\bar p\in\mathbb{R}_{\ge0}$ such that, for any input tuple $(v_t,p_t,\cvec_t)\in\I$ in the support of $\distr$, it holds $v_t\le\bar v$ and $p_t\le\bar p$.
\end{assumption}

Let $f_t:[0,1]\to\mathbb{R}$ be the reward function at time $t$ such that, for each $x\in\X_t$, \begin{equation}\label{eq:utility function}
f_t(x)\defeq (v_t-p_t)x.
\end{equation}
Then, given a generic concave regularizer $R:\mathbb{R}_{\ge0}^m\to\mathbb{R}$, we are interested in finding an online algorithm guaranteeing performance \emph{close} to that
of the optimal solution in hindsight, which we denote by $\opt(\mathscr{P})$. In particular, this baseline is the expected reward that the bidder would achieve by computing the optimal allocation $\xvec$ when the input sequence from $t=1$ to $T$ is known in advance. 
Computing $\opt(\mathscr{P})$ amounts to solving the offline optimal allocation problem for each possible realized input sequence, and by taking expectations according to the realization probability specified by $\distr$.
Formally, it is defined as follows
\begin{equation}\label{eq:opt problem}
    \opt(\mathscr{P})\defeq\expe_{\distr}\mleft[\hspace{-1.25mm}\begin{array}{l}
        \displaystyle
        \max_{\xvec:x_t\in\X_t}
        \sum_{t=1}^T f_t(x_t) + T R\mleft(\frac{1}{T}\sum_{t=1}^T \cvec_t x_t\mright)\\ [2mm]
        \displaystyle \text{\normalfont s.t. } \sum_{t=1}^T p_tx_t\le T\rho
    \end{array}\mright].
\end{equation}

\section{Parity ray regularizer}\label{sec:parity ray}

For simplicity, we assume that a type specifies a convex combination over categories, that is, for any $t$, the vector of realized categories $\cvec_t$ belongs to $\mathcal{C}\defeq\mleft\{\cvec\in[0,1]^m: \sum_i c_i=1 \mright\}$. This assumption can be relaxed to a larger bounded set.
Given a target distribution over categories $\hat\xvec\in\Delta^m$, we consider a setting where the bidder is interested in steering their realized distribution of impressions towards $\hat\xvec$.
Given a time horizon $T$, a sequence of types $(\cvec_t)_{t=1}^T$, and $\hat\xvec\in\Delta^m$, the bidder's goal is to obtain a vector of allocations $\xvec\in[0,1]^{T}$ such that, for each $i\in[m]$, $\sum_{t=1}^T c_{t,i} x_t/\sum_{t=1}^Tx_t$ is \emph{close} to $\hat x_i$.
However, directly handling such realized distribution of impressions entails solving a non-convex optimization problem.
Then, a natural question is the following: \emph{is it possible to incorporate such distributional preference as a concave regularization term in the bidder's objective function?}

We propose the following regularization term as a natural solution for this problem. 
\begin{definition}[Parity ray regularizer]\label{def:parity ray}
Let $D:\Delta^m_+ \times \Delta^m_+ \to\mathbb{R}_{\ge0}$. Then, given a target distribution $\hat\xvec\in\Delta^m$, the \emph{parity ray regularizer} is a function $R:\Delta^m_+ \times \Delta^m_+ \to\mathbb{R}_{\le0}$ such that, for any vector $\bar\xvec\in\Delta^m_+$, it holds 
\[
R(\bar\xvec)\defeq - \min_{\gamma \in [0,1]} D\mleft(\gamma\hat\xvec;\bar\xvec\mright).
\]
\end{definition}
We say that a bidder is optimizing for their \emph{parity-regularized utility} if they are solving Problem~\eqref{eq:opt problem} and $R$ is defined as in~\cref{def:parity ray}.
Intuitively, $D$ acts as a pseudo-distance measure on $\Delta^m_+$. It is not a true distance metric because we will not be assuming that it satisfies all the conditions of a metric. It can be thought of as analogous to a divergence from statistics, although we measure distances on the full-dimensional simplex. Given some current distribution $\bar \xvec$, $R$ performs a projection onto the line segment generated by $[0,1] \times \hat\xvec$, where the projection is in terms of $D$ (this is analogous to, e.g., a Bregman projection).

The ``parity ray'' name comes from the fact that we can think of $R$ as measuring the distance between a given point $\bar \xvec \in \Delta^m_+$ to the ray $\{\gamma \hat\xvec : \gamma \ge 0\}$. In our definition we only consider a segment of this ray by restricting $\gamma \in [0,1]$. This restriction is mostly for analytical convenience when deriving Lipschitz constants below. For a suitable $D$ that extends to the positive orthant, one could remove the restriction and measure distances on the whole positive orthant.

We will assume that $D$ satisfies a number of properties that will allow us to optimize it as part of an online optimization procedure.
\begin{assumption}\label{assumption dgf}
    We assume that $D(\bar \xvec; \xvec)$ satisfies the following conditions:
    \begin{itemize}
        \item $D$ is jointly convex on $\Delta^m_+ \times \Delta^m_+$, and $D(\xvec; \bar \xvec) \in [0,\bar r]$ for all $\xvec,\bar\xvec \in \Delta^m_+$.
        \item For a fixed $\gamma$, the function $D(\gamma\hat\xvec; \cdot )$ is $L$-Lipschitz with respect to the second argument under some $\ell_p$ norm.
    \end{itemize}
\end{assumption}

\begin{restatable}{proposition}{propRegularizer}\label{assumption regularizer}
    Given \cref{assumption dgf}, the regularizer $R$ is concave, closed, and for each $\xvec\in\Delta_+^m$, $R(\xvec) \in [\underline r, 0]$ with $\underline r = -\bar r$. Furthermore, $R$ is $L$-Lipschitz continuous under the $\ell_p$ norm.
\end{restatable}

For any $\ell_p$ norm, the distance $D_p(\xvec;\bar\xvec) = \|\xvec - \bar \xvec\|_p$ satisfies \cref{assumption dgf}.
In particular, $D_p$ is $1$ Lipschitz with respect to $\ell_p$, and the remaining properties follow directly from the definition of a norm and the fact that $D_p$ is defined on a bounded set $\Delta^m_+$. 
Moreover, the $\ell_2$ norm yields a simple closed form solution for the parity ray regularizer, which is particularly appealing for practical purposes. In particular, we have that for any $\bar \xvec\in\Delta_+^m$, $R(\bar \xvec)= -\|\bar \xvec- \mleft(\langle\bar\xvec,\hat\xvec\rangle / \|\hat\xvec\|_2^2\mright) \hat\xvec\|_2$, which is well defined with respect to \cref{def:parity ray} since $0\le \langle\bar\xvec,\hat\xvec\rangle / \|\hat\xvec\|_2^2\le 1$.

Another source of appropriate choices for $D$ are $f$-divergences: $D_f(\xvec;\bar\xvec) = \sum_{i=1}^m \bar \xvec_i f(\xvec_i / \bar\xvec_i)$ defined with respect to some convex function $f$, since $f$-divergences are jointly convex.
However, $f$-divergences are not necessarily $L$ Lipschitz for a fixed $L$. In the case where we wish to use a non-Lipschitz $f$, we have to modify $f$ such that $|\frac{s}{t^2} f'(s/t)| \leq L$ for $s,t\in [0,1]$.
For example, setting $f(t) = t\log t$ yields the Kullback-Leibler (KL) divergence. However, we need to slightly modify the KL divergence in order to ensure that \cref{assumption dgf} is satisfied: 
(1) we must shift if by $\log m$ to make it nonnegative, 
and (2) we make it Lipschitz in the second argument by setting $f(t) = t\log \min(t, 1/\epsilon)$; $\epsilon$ then more-or-less corresponds to a lower probability threshold at which we stop increasing the penalty. This lower threshold may actually be useful in practice, to prevent the algorithm from reacting too aggressively to allocation imbalances at the first few iterations.

\section{Online allocation algorithm}

First, we show that a suitably defined Lagrangian dual function provides an upper bound on $\opt(\distr)$ (omitted proofs are reported in~\cref{sec:omitted}).
Let $f^\ast_t:\mathbb{R}_{\ge0}\to\mathbb{R}$ be the conjugate function of $f_t$ as defined in~\cref{eq:utility function} restricted to $\X_t$, that is, for all $\mu\in\mathbb{R}_{\ge0}$,
\[
f^\ast_t(\mu)\defeq \sup_{x\in\X^t}\mleft\{f_t(x)-\mu x\mright\}.
\]
Moreover, let $R^\ast:\mathbb{R}^m\to\mathbb{R}$ be the conjugate function of the regularizer $R$ restricted to $\Delta_+^m$. Formally, for each $\lambdavec\in\mathbb{R}^m$,
\[
R^\ast(-\lambdavec)\defeq \sup_{\bar\xvec\in\Delta^m_+}\mleft\{ R(\bar\xvec)+\langle\lambdavec,\bar \xvec\rangle\mright\}.
\]

Let $\dom$ be the set of dual variables for which the conjugate of the regularizer is bounded. Formally,
\begin{restatable}{lemma}{domainDual}\label{lemma:domain}
    $\dom\defeq\mleft\{(\mu,\lambdavec)\in\mathbb{R}_{\ge0}\times \mathbb{R}^m: -\infty<R^\ast(-\lambdavec)<+\infty\mright\}=\mathbb{R}_{\ge0}\times \mathbb{R}^m$.
\end{restatable}

For a given distribution $\distr$ over the finite set of possible input tuples $\I$, the Lagrangian dual function $D(\mu,\lambdavec|\distr):\dom\to\mathbb{R}$ is such that, for each pair of dual variables $(\mu,\lambdavec)\in\dom$,
\begin{equation}\label{eq:dual_function}
  D(\mu,\lambdavec|\distr)\defeq\expe_{(v_t,p_t,\cvec_t)\sim\distr}\mleft[ f^\ast_t\mleft(\mu p_t+\langle\lambdavec,\cvec_t\rangle\mright)+R^\ast(-\lambdavec)\mright]+\rho\mu.
\end{equation}

Then, the following holds. 
\begin{restatable}{theorem}{ubOpt}\label{ub_opt}
Given a time horizon $T\ge0$, for any $(\mu,\lambdavec)\in\mathcal{D}$, $\opt(\distr)\le T D(\mu,\lambdavec|\distr)$.
\end{restatable}

Given this intermediate result, we describe an online allocation algorithm (Algorithm~\ref{alg:omd}) based on the online mirror descent scheme by~\citet{balseiro2020best}, which we adapt to account for parity-regularized objectives (which are non-separable).
We exploit the specific structure of our problem to recover better constants in the regret upper bound (see~\cref{thm:regret}), and to relax some of the assumptions required by~\citet{balseiro2020best} (see~\cref{assumption dgf}).
For a specific choice of the reference function $\psi$, Algorithm~\ref{alg:omd} gives an instantiation of the dual subgradient descent algorithm by~\citet{balseiro2020regularized} for our setting. Therefore, the dual subgradient descent algorithm by \citet{balseiro2020regularized} is a special case of Algorithm~\ref{alg:omd}.

The algorithm proceeds according to the following main steps:
\begin{description}[leftmargin=.2cm]

\item[Primal decision] 
At each time step $t$, given input tuple $(v_t,p_c,\cvec_t)$, the algorithm computes the optimal allocation $\hat x$ which maximizes an opportunity cost-adjusted reward, based on the current dual solutions $(\mu_t,\lambdavec_t)$ (Equation~\ref{eq:update_x}). The optimal decision $\hat x$ is taken if the remaining budget is enough to cover costs, that is, $p_t\hat x\le B_t$. Finally, the target distribution over categories $\bar \xvec_t$ is computed by computing the value maximizing the regularizer $R$ adjusted by an additive term accounting for the current dual solution $\lambdavec_t$ (Equation~\ref{eq:update_x_bar}).

\item[Dual variables update]
First, the algorithm computes an unbiased stochastic estimator of a subgradient of $D(\mu,\lambdavec|\distr)$ at $(\mu_t,\lambdavec_t)$ (Equation~\ref{eq:subgradient}). The algorithm employs this estimator to update the vector of dual variables by performing an online dual mirror descent descent step with step size $\eta$ and reference function $\psi$ (Equation~\ref{eq:dual update}).
\end{description}

Intuitively, at each time step $t$, the algorithm compares (i) the actual expenditure from the auction to the expected rate of spend per iteration $\rho$, and (ii) the target distribution over categories at $t$ to the realized category. 
If the actual expenditure is higher (resp., lower) than this expected rate, then the algorithm surmises that future opportunities will offer higher (resp., lower) \emph{bang for the buck} than current opportunities, and therefore the algorithm lowers (resp., raises) the bid shading multiplier.
At the same time, if the realized category causes an undesired skew, and moves the realized distribution of impressions away from the desired $\bar\xvec_t$, then the algorithm will adjust its dual solutions to penalize allocations from that category, and increase the likelihood of acceptance for items coming from \emph{under-represented} categories.

\begin{algorithm}[t]
	\SetAlgoLined
	%\LinesNumbered
	\KwIn{Initial dual solutions $(\mu_1,\lambdavec_1)$, time horizon $T$, initial budget consumption $B_1=B$, function $\psi$, step-size $\eta$.}
	
	\For{$t=1,\ldots,T$}{
	\text{\normalfont Observe} $(v_t,p_t,\cvec_t)\sim\distr$\\
	\text{\normalfont Compute primal decision:}\\
	\begin{equation}\label{eq:update_x}
    \hat x \in\argmax_{x\in\X_t}\mleft\{
    f_t(x)-\mu_t p_t x-\langle \lambdavec_t, \cvec_t\,\, x\rangle
    \mright\}
    \end{equation}\\
    \begin{equation*}
    x_t=\mleft\{\hspace{-1.25mm}\begin{array}{l}
        \hat x \hspace{.3cm} \text{\normalfont if } p_t\hat x\le B_t\\ [2mm]
        0 \hspace{.4cm} \text{\normalfont otherwise}
    \end{array}\mright.
    \end{equation*}\\
    \begin{equation}\label{eq:update_x_bar}
    \bar\xvec_t \in\argmax_{\bar\xvec'\in\Delta_+^m}\mleft\{
    R(\bar\xvec')+\langle\lambdavec_t, \bar\xvec'\rangle
    \mright\}
    \end{equation} \\
    \text{\normalfont Update resource consumption:} $B_{t+1}= B_t-p_t x_t$\\
    \text{\normalfont Compute subgradient $\subvec\in\mathbb{R}^{m+1}$ of} $D(\mu_t,\lambdavec_t|\distr)$:\\
    \begin{equation}\label{eq:subgradient}
        % g_j=\mleft\{\hspace{-1.25mm}\begin{array}{l}
        % \rho - p_t \hat x \hspace{1cm} \text{\normalfont if } j=0\\ [2mm]
        % \bar x_{t,j} - b_{t,j} \hat x \hspace{.4cm} \text{\normalfont otherwise}
        % \end{array}\mright.
        \subvec=\mleft(\rho-p_t\hat x,\,\, \bar\xvec_{t}-\cvec_t\hat x \mright)
    \end{equation}\\
    \text{\normalfont Update dual variables:}\\
    \begin{equation}\label{eq:dual update}
    (\mu_{t+1},\lambdavec_{t+1})=\argmin_{(\mu,\lambdavec)\in\dom}\mleft\{ \langle \subvec,(\mu,\lambdavec)\rangle + \frac{1}{\eta} B_{\psi}((\mu,\lambdavec),(\mu_t,\lambdavec_t))\mright\}
    \end{equation}
    }
	\caption{Online allocation algorithm for parity-regularized pacing}
	\label{alg:omd}
\end{algorithm}

\subsection{Regret bound}

We focus on the stochastic setting in which, at each $t$, an input tuple is drawn i.i.d. from $\distr$. This setting is particularly relevant for applications to large Internet advertising platforms in which the number of bidders interacting makes the environment unlikely to react adversarially to choices of one specific bidder. We show that Algorithm~\ref{alg:omd} attains a sublinear regret of order $O(T^{1/2})$.

Let $\tau$ be the time at which the budget of the bidder is depleted when following Algorithm~\ref{alg:omd}. By playing according to Algorithm~\ref{alg:omd}, the bidder attains the following expected reward:
\begin{equation}\label{eq:expected_rev}
    \rew(\distr)\defeq\expe_{\distr}\mleft[ \sum_{t=1}^T f_t(x_t)+TR\mleft(\frac{1}{T}\sum_{t=1}^T\cvec_t x_t\mright)\mright].
\end{equation}

We start with the following intermediate result which provides a lower bound on the expected reward attained by following Algorithm~\ref{alg:omd} up to the time $\tau$ in which the budget is fully depleted. 
\begin{restatable}{lemma}{lbRew}\label{lemma:lb_rew}
    Consider an arbitrary $\distr\in\Delta^{\I}$. For each $t\in[T]$, let $x_t$ and $\bar x_t$ be computed according to~\cref{eq:update_x} and~\cref{eq:update_x_bar}, respectively. Moreover, given the time of budget depletion $\tau$, let
    \[
    \bar\mu=\frac{1}{\tau}\sum_{t=1}^\tau \mu_t \quad\text{\normalfont and }\quad \bar\lambdavec=\frac{1}{\tau}\sum_{t=1}^\tau \lambdavec_t.
    \]
    Then, it holds
    \[
    \expe_\distr\mleft[ \sum_{t=1}^\tau\mleft( f_t(x_t) + R(\bar \xvec_t)\mright)\mright]\ge \expe_{\distr}\mleft[\tau D(\bar\mu,\bar\lambdavec|\distr)-\sum_{t=1}^\tau\mleft(\mu_t(\rho-p_tx_t)-\langle\lambdavec_t,\bar\xvec_t - \cvec_tx_t\rangle\mright)\mright].
    \]
\end{restatable}

In order to prove the regret bound of~\cref{thm:regret} we need the following additional assumption on the reference function $\psi$ of the Bregman divergence $B_\psi$.
\begin{assumption}\label{assumption reference function}
The reference function $\psi$ is $\sigma$-strongly convex with respect to an $\ell_p$ norm $\norm{\cdot}_p$.
\end{assumption}
This assumption is a relaxation of the requirements in~\citet[Assumption 2]{balseiro2020best} (i.e., strong convexity with respect to $\|\cdot\|_1$). Indeed, by exploiting the specific structure of our problem, we can prove the following.

\begin{restatable}{theorem}{regretBound}\label{thm:regret}
Consider Algorithm~\ref{alg:omd} with step-size $\eta\ge 0$ and initial solution $(\mu_1,\lambdavec_1)\in\dom$. Suppose~\cref{assumption:upper bounds}, \cref{assumption regularizer}, and \cref{assumption reference function} are satisfied, and the requests are drawn i.i.d. from an unknown distribution $\distr\in\Delta^\I$. Then, for any $T\ge 1$, it holds
\[
\opt(\distr)-\rew(\distr)\le C_1 + \frac{G^2\eta}{\sigma} T +\frac{1}{\eta}C_2,
\]
where $C_1=(\bar v -\underline{r} + 2L)\bar p/\rho$, $G=\max\{\rho+\bar p,2\}$, and \[C_2=2\sup\{B_\psi((\mu,\lambdavec),\bb{d}_1):(\mu,\lambdavec)\in\dom,\norm{\lambdavec}\le L\}.\]
\end{restatable}
By setting $\eta\sim T^{-1/2}$ we obtain a regret of order $O(T^{1/2})$, which is the optimal order of regret that can be attained in this setting as shown by~\citet[Lemma 1]{arlotto2019uniformly}.

\section{Regularized pacing within a paced-auction framework}

One of the major selling points of our approach to enforcing distributional preferences via regularization is that Algorithm~\ref{alg:omd} can easily be incorporated in existing pacing systems that are not necessarily built around this objective.
In Section~\ref{sec:pacing in auctions} we described how pacing systems with proxy bidders work for large-scale Internet advertising. In that setup, the platform controls a proxy bidder on behalf of every advertiser $i$, where the proxy bidder uses a control algorithm on the pacing parameter $\alpha_i$ used to construct bids, in order to ensure that the advertiser satisfies their budget constraint. An important feature of this setup is that each proxy bidder can control their own pacing parameter $\alpha_i$ purely through their observed spending. This leads to a highly decentralized framework, where communication between proxy bidders is only necessary via the second-price auctions that sell items whenever they show up, and bidders only need to submit their paced bids to these auctions. This decentralization is important, because centralized optimization problems may not be realistic to solve every time a user shows up on the site.
In the remainder of the section we explain how a bidder may implement our Algorithm~\ref{alg:omd} via a form of decentralized double-pacing. We do this to explain it in generality, but we emphasize that the proxy-bidder setting is one of our main motivations.

Consider a bidder who attempts to maximize parity-regularized utility. By using Algorithm~\ref{alg:omd}, the bidder will have two parameters at time $t$: $\mu_t$ and $\lambdavec_t$, with some item of value $v_{t}$ having appeared (the price $p_t$ would not be known yet since the auction has not been run). Algorithm~\ref{alg:omd} requires that the bidder achieve an allocation that solves \cref{eq:update_x}, which we can rewrite as follows:
\begin{align*}
    \argmax_{x\in\{0,1\}}\mleft\{
    f_t(x)-\mu_t p_t x-\langle \lambdavec_t, \cvec_tx\rangle
    \mright\}
    &=
    \argmax_{x\in\{0,1\}}\mleft\{
    v_{t}x- (1+\mu_t) p_t x-\langle \lambdavec_t, \cvec_t\rangle x
    \mright\} \\
    &=
    \argmax_{x\in\{0,1\}}\mleft\{
    (v_{t} - \langle \lambdavec_t, \cvec_t\rangle)x- (1+\mu_t) p_t x
    \mright\} \\
    &=
    \argmax_{x\in\{0,1\}}\mleft\{
    \frac{v_{t} - \langle \lambdavec_t, \cvec_t\rangle}{1+\mu_t} x -  p_tx
    \mright\}.
\end{align*}

From the above derivations, we observe at each $t$ that there are three possible cases: (i) the bidder needs to win the item at time $t$ if $\frac{v_{t} - \langle \lambdavec_t, \cvec_t\rangle}{1+\mu_t} > p_t$; (ii) the bidder has to pass on the item if the strict inequality is reversed; (iii) the bidder is indifferent if equality holds.
Even though the bidder does not know $p_t$, this is exactly achieved within a second-price auction framework by bidding the value $\frac{v_{t} - \langle \lambdavec_t, \cvec_t\rangle}{1+\mu_t}$ at time $t$. Here, $1/(1+\mu_t)$ plays the role of the standard multiplicative pacing parameter, but before multiplicative pacing, the bidder first applies an additive pacing term $\langle \lambdavec_t, \cvec_t\rangle$.

Thus, we see that our approach can easily be adopted by individual bidders that wish to incorporate distributional regularization as part of their bidding procedure in a repeated-auction setting.
For the same reason, proxy bidders can easily incorporate this in the case where an advertising platform wishes to enforce some amount of distributional parity while performing budget-pacing on behalf of the advertisers.

\begin{figure}[h]
\centering
\begin{minipage}{.45\textwidth}
\includegraphics[scale=0.28]{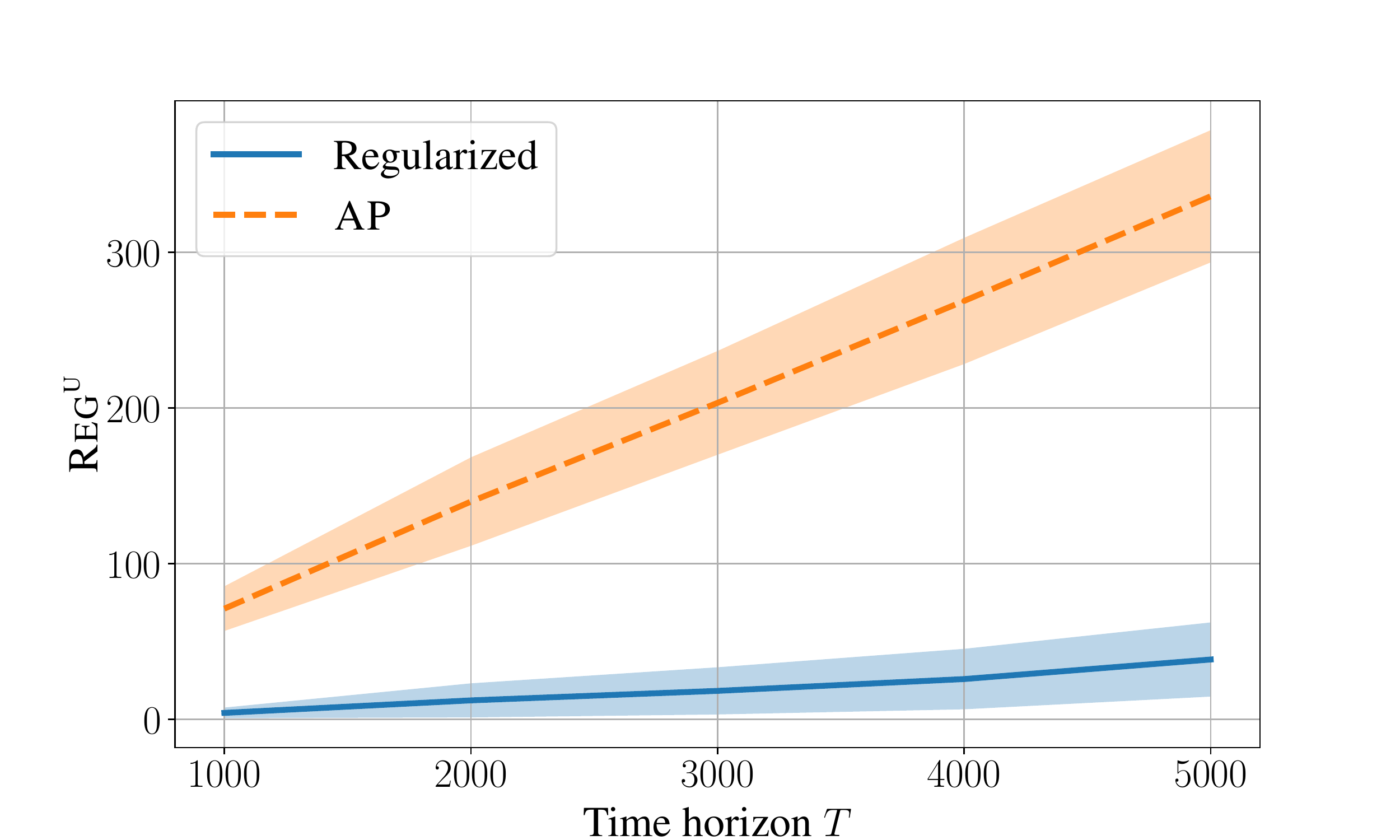}
\end{minipage}
\hspace{.5cm}
\begin{minipage}{.45\textwidth}
\includegraphics[scale=0.28]{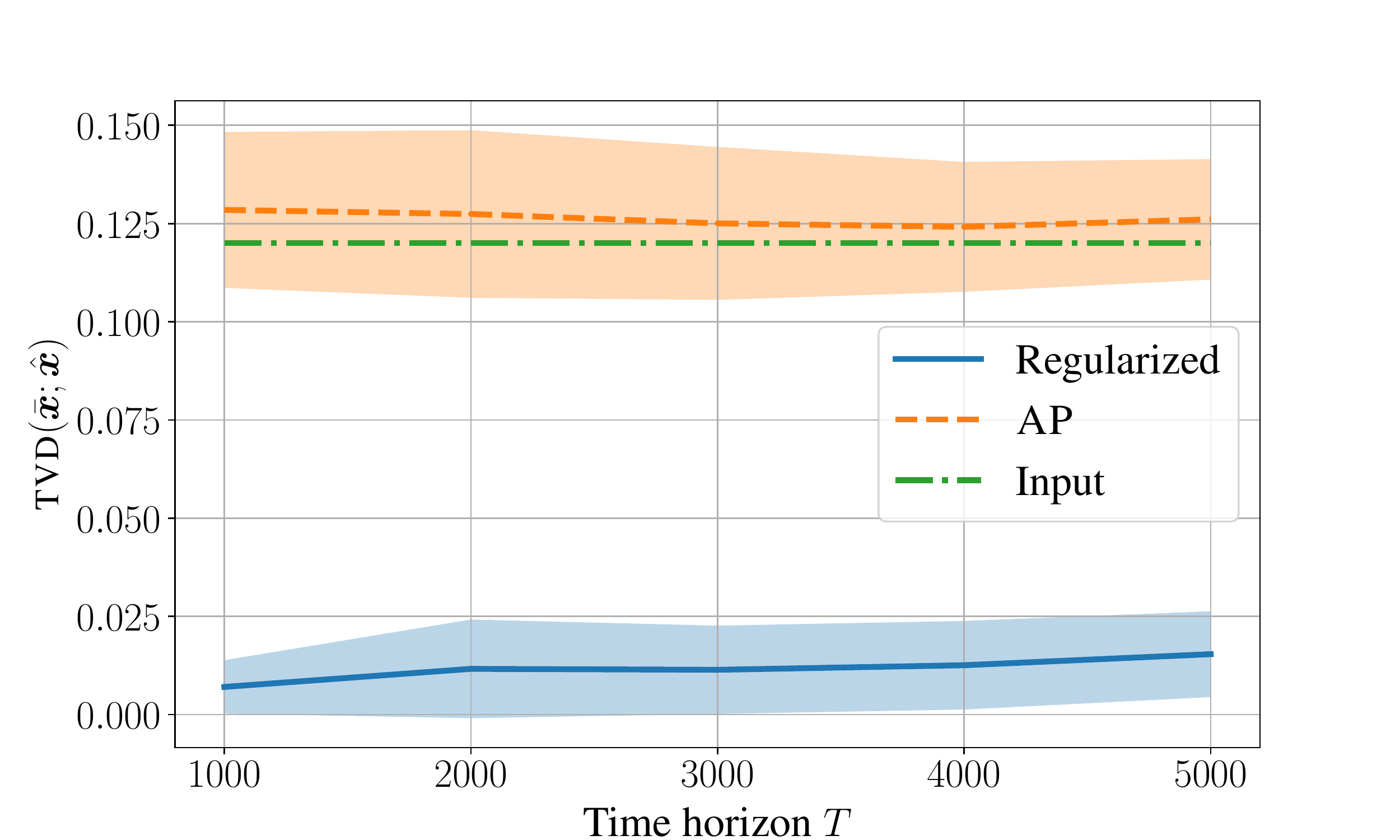}
\end{minipage}

\begin{minipage}{.45\textwidth}
\includegraphics[scale=0.28]{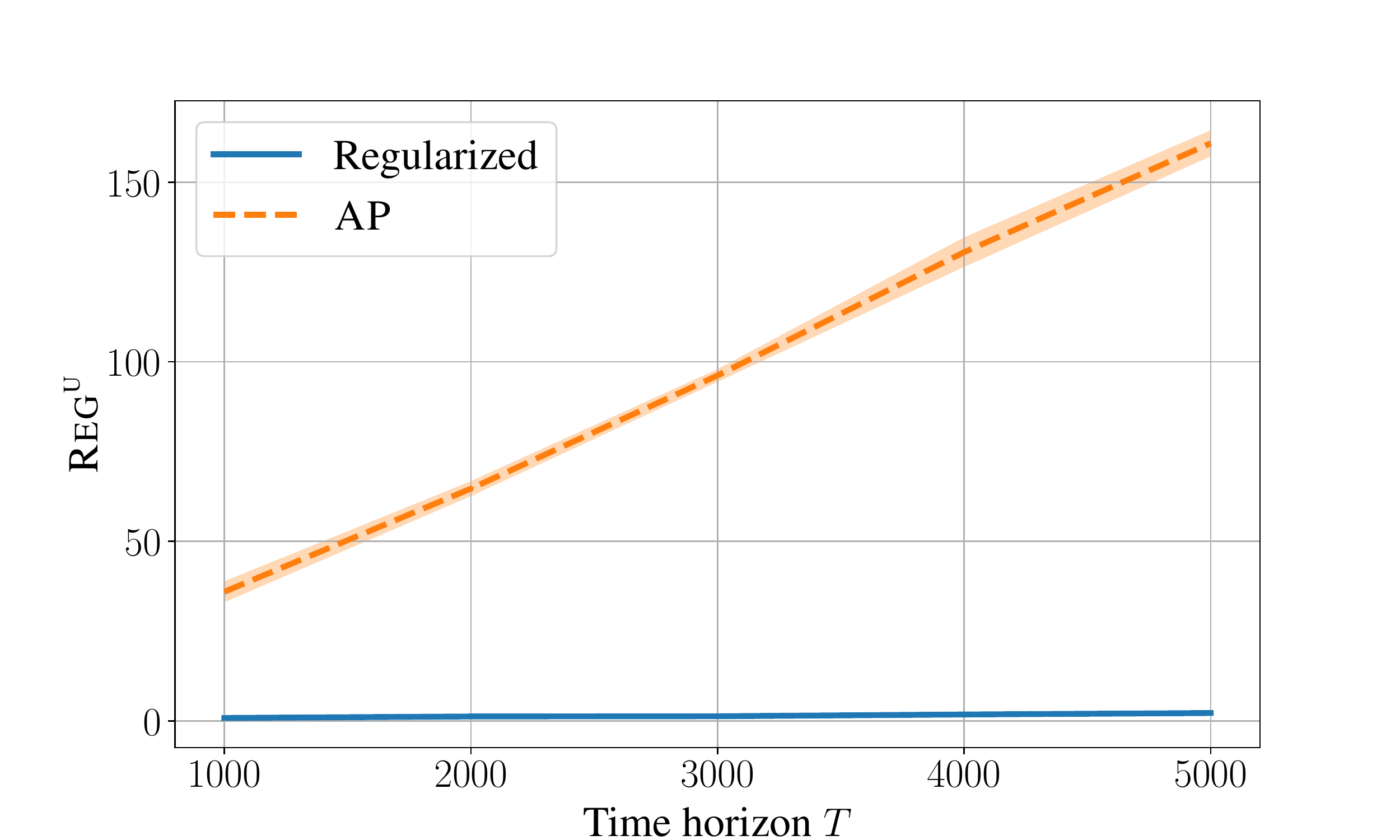}
\end{minipage}
\hspace{.5cm}
\begin{minipage}{.45\textwidth}
\includegraphics[scale=0.28]{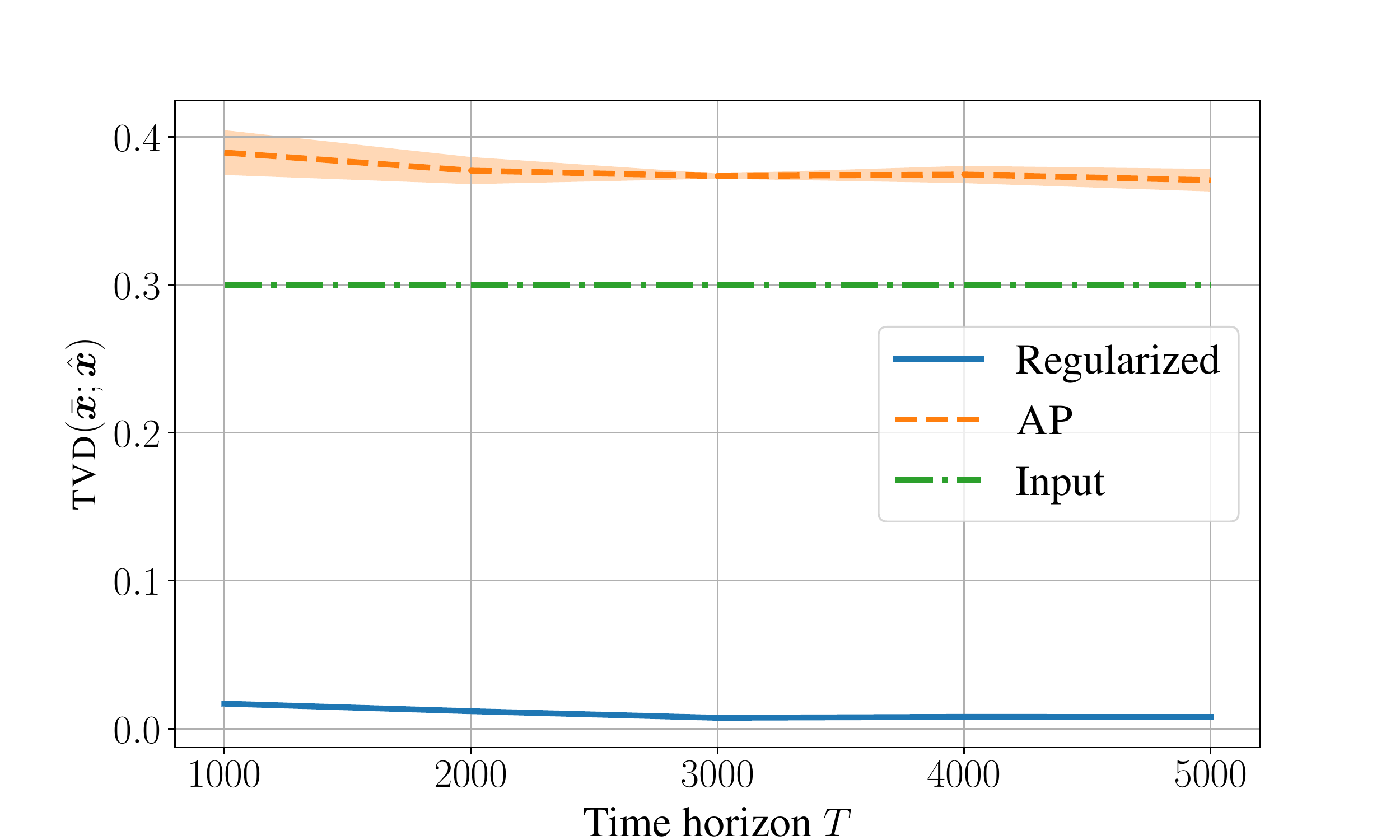}
\end{minipage}
\caption{\footnotesize \textbf{(Top-Left)}: Empirical upper bound on the regret (i.e., $\reg^\textsc{u}$ as defined in~\cref{eq:regret ub}) for breakdown \textsc{B2}. \textbf{(Top-Right)}: Total variation distance from the target distribution $\hat\xvec$ for \textsc{B2}. \textbf{(Bottom-Left)}: Empirical upper bound on the regret for breakdown \textsc{B3}. \textbf{(Bottom-Right)}: Total variation distance from the target distribution $\hat\xvec$ for \textsc{B3}. \emph{Input} denotes the total variation distance between $\hat\xvec$ and the underlying empirical category distribution computed on the dataset.}
\label{fig:exp}
\end{figure}

\section{Experimental evaluation}\label{sec: exp eval}

In this section, we present numerical experiments on data from a large Internet advertising company.

\paragraph{Experimental setup}

We construct a real-world dataset through logs of a large Internet advertising company. In this dataset, we have one advertiser (more specifically, we consider data at individual campaign level) participating in a sequence of auctions. We consider only ad requests coming from mobile devices. For each ad request we retrieve the advertiser's bid, the price computed via the auction mechanism, and the relevant categories for the request.
We observe that, due to the nature of the bidding system, the bids provide a good approximation to the true advertiser's valuations.
Bids and prices are suitably normalized for privacy reasons. 
Then, we generate a dataset consisting of $10^6$ entries (i.e., one entry per ad requests) and, at each run of the algorithm, we sample from this dataset $T$ requests without replacement.
We consider two different ad-request breakdowns. The first one, denoted as breakdown \textsc{B2}, is a binary breakdown (i.e., $m=2$) where the category is the \emph{type} of mobile device originating the request (we consider the two most frequent such types). The second one, which we denote by \textsc{B3}, is a breakdown into three categories (i.e., $m=3$) determined by the manufacturer of the mobile device.
In all of the experiments we set $\rho$ to be approximately equal to the $0.5$ quantile of the empirical distribution of normalized prices from our dataset. In this way, the advertiser is set to win at most half of the ad requests which they are presented.

\paragraph{Algorithms} 
We instantiate parity ray regularizers with $D_2(\xvec;\bar\xvec)=\|\xvec-\bar\xvec\|_2$ as the pseudo-distance measure.
Moreover, given a sequence of input tuples $(v_t,p_t,\cvec_t)_{t=1}^T$, let 
\begin{equation*}
        \opt^\textsc{u}\mleft((v_t,p_t,\cvec_t)_{t=1}^T\mright)\defeq\max_{\xvec:x_t\in\X_t}
        \sum_{t=1}^T f_t(x_t)\,\,\text{\normalfont s.t. }\,\sum_{t=1}^T p_tx_t\le T\rho
\end{equation*}
be the value of an optimal allocation in absence of regularization penalties. Then, we evaluate the regret guarantees of our algorithm by computing an upper bound to the parity-regularized regret $\opt(\distr)-\rew(\distr)$. In particular, we compute
\begin{equation}\label{eq:regret ub}
\reg^\textsc{u}\mleft((x_t,v_t,p_t,\cvec_t)_{t=1}^T\mright)\defeq \opt^\textsc{u}\mleft((v_t,p_t,\cvec_t)_{t=1}^T\mright)-\rew\mleft((x_t,v_t,p_t,\cvec_t)_{t=1}^T\mright),
\end{equation}
where $\rew\mleft((x_t,v_t,p_t,\cvec_t)_{t=1}^T\mright)$ is the realized reward for the sequence of online decisions $\xvec$.
The baseline algorithm which we employ is the \emph{adaptive pacing algorithm} (\textsc{AP}) by~\citet{balseiro2019learning}, which guarantees optimal regret in the stochastic setting without regularization.
The step size is $\eta=0.1/\sqrt{T}$. This choice guarantees the optimal order of regret both for our algorithm, and for adaptive pacing (in the non-reqularized setting). Dual variables are initialized to $0$.
Convex programs are solved using the Xpress Optimization Suite 8.10~\cite{ficosuite}. 
All experiments are run on a 24-core machine with 57 GB of RAM.

\paragraph{Results}

We run $20$ independent trials for each time horizon $T\in\{10^3,2\cdot 10^3,\ldots,5\cdot10^3\}$. 
The left column of~\cref{fig:exp} reports the empirical upper bounds on the regret computed as in~\cref{eq:regret ub}. The right column of~\cref{fig:exp} reports the \emph{total variation distance} between the realized distributions of impressions $\bar\xvec$ and the target distribution $\hat\xvec$. We compute this distance as $\textsc{tvd}(\bar\xvec;\hat\xvec)=\sup_j|\bar x_j-\hat x_j|$. 
The first row of~\cref{fig:exp} reports result for the \textsc{B2} breakdown. In this setting, we set $\hat\xvec=(.5,.5)$, that is, the advertiser seeks to reach a uniform realized distribution of impressions over the two mobile device types.
Algorithm~\ref{alg:omd} achieves small regret over all time horizons, and guarantees a realized distribution of impression close to the target. On the other hand, the adaptive pacing algorithm shows linear regret in $T$, which is due to the penalty for not being able to steer the realized distribution of impression towards the target. 
The same behavior can be observed in the \textsc{B3} setting. Here, the empirical distribution of categories of the dataset is roughly $(.5,.3,.2)$. To make the problem harder, we try to skew the realized distribution of impressions towards the least frequent category by setting $\hat\xvec=(.1,.3,.6)$. Algorithm~\ref{alg:omd} manages to keep the resulting distribution close to the target even in this challenging setting. It is interesting to notice that the distribution over categories of the input data would be closer to the target distribution than what achieved by the adaptive pacing algorithm.
The same behavior can be observed when we artificially create an unbalance between the distribution of prices for the different categories. In these extreme cases, the adaptive pacing algorithm forces the advertiser to win only impression of the \emph{cheap} category (see Appendix~\ref{sec:exp appendix}).

\small{
\bibliographystyle{plainnat}
\bibliography{refs}
}
%%%%%%%%%%%%%%%%%%%%%%%%%%%%%%%%%%%%%%%%%%%%%%%%%%%%%%%%%%%%
%%%%%%%%%%%%%%%%%%%%%%%%%%%%%%%%%%%%%%%%%%%%%%%%%%%%%%%%%%%%
\clearpage
\appendix
\section{Related works}\label{sec:related works}

Online allocation problems have been studied from a number of different perspectives.
One stream of works addresses the case in which inputs are chosen by an adversary (see, e.g.,~\citet{mehta2007adwords,feldman2009online}).
In our work, motivated by Interned advertising applications, we study the stochastic case in which inputs are drawn i.i.d. from an unknown distribution. Stochastic input models have been studied, among others, by~\citet{devanur2009adwords,mirrokni2012simultaneous,devanur2012asymptotically,mahdian2012online,feldman2009onlineb,goel2008online}.

In the setting of stochastic input models with linear reward functions it is possible to guarantee an optimal order of regret of $O(T^{1/2})$~\cite{agrawal2014dynamic,devanur2019near}. 
More recent works by~\citet{li2020simple} and~\citet{balseiro2020dual} propose simple algorithms with $O(T^{1/2})$ regret guarantees that, in contrast with previous works, do not require periodically solving large linear programs.
Finally, under some specific structural assumptions, it is possible to obtain a regret bound of order $O(\log T)$~\cite{li2019online,jasin2015performance}.

As discussed in~\cref{intro} and~\cref{sec:parity ray}, we study how to optimize for parity-regularized utilities, which are non-separable in nature. 
Two related works addressing similar problems are the following:
\begin{itemize}
\item~\citet{agrawal2014dynamic} focus on how to solve general online stochastic convex programs that allow general concave objectives and convex constraints.
However, their approach requires periodically solving convex optimization programs on historical data, unless the optimal value of the objective is known. Moreover, their algorithm handles resource constraints as \emph{soft} constraints, while we are interested in enforcing a hard constraint on the budget.

\item~\citet{balseiro2020regularized} propose a dual subgradient descent algorithm for a general class of non-separable objectives. 
In our work, we focus on the specific real-world problem of modeling advertisers' distributional preferences within standard pacing systems.~\citet{balseiro2020regularized} focus on a more general class of problems, without providing a concrete way to model our objective. 
In particular, we observe that the regularizer for \emph{max-min fairness} which the authors propose in their Example 2 is not suitable for our setting. In order to employ their regularizer to model distributional preferences, an advertisers would have to appropriately pick parameters $\lambda$ and $\rho_i$ depending on the contingent state of the environment during the $T$ time horizon. Specifically, this choice would require knowing beforehand the underlying input distribution of requests, as well as the \emph{scale} of the number of ad opportunities that the advertiser will win during the $T$ steps. 
On the other hand, the parity ray regularizer simplifies the interaction between the advertiser and the pacing mechanism, as the advertiser just needs to specify their target distribution, without requiring any additional knowledge on the current state of the environment. Moreover, the parity ray regularizer can be equipped with any pseudo-distance measure satisfying~\cref{assumption dgf}.
Finally, in our paper we adapt the online mirror descent scheme by~\citet{balseiro2020best} to the case of non-separable objectives (which allows for a more general class of reference functions than the online gradient descent scheme of \citet{balseiro2020regularized}). The dual-descent scheme by~\citet{balseiro2020regularized} can be derived from our OMD scheme through an appropriate choice of the reference function. 
Our analysis draws heavily on the ideas developed in both~\citet{balseiro2020best} and~\citet{balseiro2020regularized}.
By exploiting the specific structure of our problem we show how to relax some of the assumptions required by~\citet{balseiro2020best} (see, .e.g., Assumption 2 of their paper). 
\end{itemize}

There is also a stream of literature studying budget management problems within the framework of stochastic bandits with knapsacks constraints (see, e.g., the work by \citet{avadhanula2021stochastic}).
The difference between our setting and the bandits literature is that, in online allocation problems, the advertiser observes the reward function and consumption matrix \emph{before} making a decision, while in the bandit setting those are revealed \emph{after} each decision is made.
In order to deal with unknown reward functions and consumption matrices, the most frequent approach in contextual bandits is to discretize the context and action space. However, this inevitably leads to worse performance, and does not fit well with real-world pacing systems such as those described in \cref{sec:pacing in auctions}. As an example, in Internet advertising applications bandits algorithms typically attain regret on the order of  $O(T^{3/4})$~\cite{agrawal2016efficient,badanidiyuru2014resourceful}.

Another line of research studies how individual bidders should optimize their budget spending across a set of auctions. This problem has been cast, for example, as a knapsack problem~\cite{feldman2007budget,borgs2007dynamics,zhou2008budget}, a Markov Decision Process~\cite{gummadi2013optimal,amin2012budget,wu2018budget}, a constrained optimization problem~\cite{zhang2012joint,zhang2014optimal}, and an optimal control problem~\cite{xu2015smart,zhang2016feedback}.

In a recent work,~\citet{nasr2020bidding} study how to compute bidding strategies for advertisers which aim at satisfying various types of parity constraints with respect to a binary user breakdown. 
Their model is different from ours since (i) we are interested in letting advertisers reach an arbitrary target distribution over an arbitrary user breakdown, while their goal is guaranteeing that specific parity constraints are satisfied for a binary user breakdown; (ii) we have a hard constraint on the budget, while they assume advertisers have an unlimited budget; (iii) the techniques which are adopted are very different since the authors model the problem as an MDP which is then solved with an ad-hoc method to properly manage the constraints, while we address distributional preferences directly within the pacing architecture.

\section{Omitted proofs}\label{sec:omitted}

\propRegularizer*
\begin{proof}
    Concavity follows from the facts that $D$ is jointly convex and convexity is preserved by taking the infimum over one argument, finally taking the negative leads to concavity.
    The fact that $R$ is $L$ Lipschitz follows since $D$ is $L$ Lipschitz in the second argument for each $\gamma$, and infimum preserves Lipschitz continuity when all functions in the infimum have the same Lipschitz constant.
    Closedness is implied by $L$ Lipschitzness, since closedness of a convex function is equal to lower-semicontinuity, a weaker condition that Lipschitz continuity.
\end{proof}

\domainDual*
\begin{proof}
Given $\lambdavec\in\mathbb{R}^m$, the following holds
\begin{align*}
    R^\ast(-\lambdavec)&=\sup_{\bar\xvec\in\Delta^m_+}\mleft\{ R(\bar\xvec)+\langle\lambdavec,\bar \xvec\rangle\mright\} \\ 
    &\le  \sup_{\bar\xvec\in\Delta^m_+} R(\bar\xvec)+\sup_{\bar\xvec\in\Delta^m_+}\langle\lambdavec,\bar \xvec\rangle \\
    &\le ||\lambdavec||_1 < +\infty,
\end{align*}
where the third inequality follows by Assumption~\ref{assumption regularizer}. Analogously, given $\lambda\in\mathbb{R}^m$, we have \[R^\ast(\lambdavec)\ge\underline{r}-||\lambdavec||_1>-\infty,\]
which proves the statement.
\end{proof}

\ubOpt*
\begin{proof}
Let $(\bb{v},\bb{p},\cvec)=((v_t)_{t=1}^T,(p_t)_{t=1}^T,(\cvec_t)_{t=1}^T)$ be an arbitrary sequence of input tuples. From the definition of the baseline $\opt$ we have that:
\begin{align*}
        \opt(\mathscr{P})\defeq & \expe_{(\bb{v},\bb{p},\cvec)\sim\distr}\mleft[\hspace{-1.25mm}\begin{array}{l}
        \displaystyle
        \max_{\xvec:x_t\in\X_t}
        \sum_{t=1}^T f_t(x_t) + T R\mleft(\frac{1}{T}\sum_{t=1}^T \cvec_t x_t\mright)\\ [2mm]
        \displaystyle \text{\normalfont s.t. } \sum_{t=1}^T p_tx_t\le T\rho\end{array}\mright]\\
        = & 
        \expe_{(\bb{v},\bb{p},\cvec)\sim\distr}\mleft[\hspace{-1.25mm}\begin{array}{l}
        \displaystyle
        \max_{\substack{\xvec:x_t\in\X_t\\ \bar\xvec \in\Delta_+^m}}
        \sum_{t=1}^T f_t(x_t) + T R\mleft(\bar\xvec\mright)\\ [3mm]
        \displaystyle \,\, \text{\normalfont s.t. } \hspace{.4cm}\sum_{t=1}^T p_tx_t\le T\rho\\[3mm]
        \displaystyle \hspace{.5cm}\hspace{.5cm}\bar\xvec=\frac{1}{T}\sum_{t=1}^T \cvec_t x_t
    \end{array}\mright].
\end{align*}

Then, by exploiting weak duality and the fact that input tuples are drawn i.i.d. from $\distr$, we have that, for any $(\mu,\lambdavec)\in\dom$,
\begin{align*}
  \opt(\mathscr{P})&\le \expe_{(\bb{v},\bb{p},\cvec)\sim\distr}\mleft[\max_{\xvec,\bar\xvec}\mleft\{ \sum_{t=1}^T f_t(x_t) + T R\mleft(\bar\xvec\mright) - \mu\sum_{t=1}^T p_tx_t + T\mu\rho + \langle\lambdavec,T\bar\xvec -  \sum_{t=1}^T \cvec_t x_t\rangle\mright\} \mright]\\
  &\le \expe_{(\bb{v},\bb{p},\cvec)\sim\distr}\mleft[ \sum_{t=1}^T f^\ast_t(\mu p_t + \langle\lambdavec,\cvec_t\rangle)+TR^\ast(-\lambdavec)+T\rho\mu\mright]\\
    &= T \,\expe_{(v_t, p_t,\cvec_t)\sim\distr}\mleft[ f^\ast_t(\mu p_t + \langle\lambdavec,\cvec_t\rangle)+R^\ast(-\lambdavec)\mright]+\rho\mu\\
  &=TD(\mu,\lambdavec|\distr).
\end{align*}
This proves our statement.
\end{proof}

\lbRew*
\begin{proof}
For each $t\le\tau$, by~\cref{eq:update_x} and~\cref{eq:update_x_bar} we have
\begin{align*}
    \expe_\distr\mleft[ f_t(x_t) + R(\bar\xvec_t) \mright] & = 
    \expe\mleft[ f_t^\ast(\mu_t p_t + \langle\lambdavec_t,\cvec_t\rangle) + \mu_t p_t x_t + \langle\lambdavec_t,\cvec_t x_t\rangle + R^\ast(\lambdavec_t) -\langle\lambdavec_t, \bar\xvec_t\rangle \mright] \\
    & = \expe\mleft[ D(\mu_t,\lambdavec_t|\distr) - \rho\mu_t   + \mu_t p_t x_t + \langle\lambdavec_t,\cvec_t x_t\rangle -\langle\lambdavec_t, \bar\xvec_t\rangle \mright] \\
    &= \expe\mleft[  D(\mu_t,\lambdavec_t|\distr) - \mu_t(\rho- p_tx_t) -\langle \lambdavec_t,\bar\xvec_t - \cvec_t x_t\rangle \mright],
\end{align*}
where the second equality holds by~\cref{eq:dual_function}.
By summing over $t=1,\ldots,\tau$ we obtain
\begin{align*}
    \expe_\distr\mleft[ \sum_{t=1}^\tau f_t(x_t) + R(\bar \xvec_t)\mright] & = \sum_{t=1}^\tau\expe\mleft[ f_t(x_t) + R(\bar \xvec_t)\mright]\\
    & =\expe\mleft[\sum_{t=1}^\tau \mleft(  D(\mu_t,\lambdavec_t|\distr) - \mu_t(\rho-\mu_t p_t) -\langle \lambdavec_t,\bar\xvec_t - \cvec_t x_t\rangle \mright) \mright] \\ 
    & \ge \expe\mleft[\tau D(\bar\mu,\bar\lambdavec|\distr)- \sum_{t=1}^\tau \mleft(\mu_t(\rho-p_tx_t)-\langle\lambdavec_t,\bar\xvec_t - \cvec_tx_t\rangle  \mright)\mright],
\end{align*}
where the last inequality holds by convexity of the Lagrangian dual function. Specifically, $f_t^\ast(\mu p + \langle\lambdavec,\cvec\rangle)$ is the composition of a convex function (by definition of conjugate function) with an affine mapping, which is convex. Therefore, $D$ is convex since it is the result of a sum of convex functions, and expectation preserves convexity.
\end{proof}

For completeness, we report the standard regret bound for OMD (see, e.g.,~\citet[Section 6]{orabona2019modern} for further details). This bound will be employed in the proof of~\cref{thm:regret}.
\begin{theorem}[Regret bound for OMD]\label{thm:omd regret}
Consider a sequence of convex functions $\xi_t(\yvec)$. Let $B_\psi$ be the Bregman divergence with respect to $\psi$, and assume $\psi$ to be $\sigma$-strongly convex with respect to $\|\cdot\|_p$. Let $\bb{g}_t\in\partial\xi_t(\yvec_t)$, and
\[
\yvec_{t+1}=\argmin_{\yvec}\langle\bb{g}_t,\yvec\rangle+\frac{1}{\eta}B_\psi(\yvec,\yvec_t).
\]
If subgradients are bounded by $\|\bb{g}_t\|_{p,\ast}\le G$, then, for any $\bb{u}\in\mathcal{V}$, the following upper bound holds
\[
\sum_{t=1}^T\mleft(\xi_t(\yvec_t)-\xi_t(\bb{u})\mright)\le 
\textsc{U}(T,\bb{u}),
\]
with
\[
\textsc{U}(T,\bb{u})\defeq\frac{G^2\eta}{2\sigma}T+\frac{1}{\eta}B_{\psi}(\bb{u},\yvec_1).
\]
\end{theorem}

This can be exploited in order to prove the following regret bound for Algorithm~\ref{alg:omd}.

\regretBound*
\begin{proof}
Recall that $\tau$ is the time at which the budget is depleted.
By~\cref{ub_opt} and since $v_t\le \bar v$ and $p_t\ge0$ we have that, for any $(\mu,\lambdavec)\in\dom$,
\begin{align*}
   \opt(\distr) & =\frac{\tau}{T}\opt(\distr)+\frac{T-\tau}{T}\opt(\distr) \\ 
   & \le \tau D(\mu,\lambdavec|\distr) + (T-\tau)\bar v.
\end{align*}
Let $\xi_{1,t}(\mu)=\mu(\rho-p_tx_t)$ and $\xi_{2,t}(\lambdavec)=\langle\lambdavec,\bar\xvec_t - \cvec_tx_t\rangle$ be the two complementary slackness terms at iteration $t$.
From~\cref{eq:expected_rev} and~\cref{lemma:lb_rew} we have
\begin{align*}
    \rew(\distr)& =\expe\mleft[ \sum_{t=1}^\tau f_t(x_t)+T R\mleft(\frac{1}{T}\sum_{t=1}^T\cvec_t x_t\mright)\mright]\\
    &\ge \expe\mleft[\tau D(\bar\mu,\bar\lambdavec|\distr)- \sum_{t=1}^\tau \mleft(\xi_{1,t}(\mu_t)+\xi_{2,t}(\lambdavec_t) + R(\bar \xvec_t) \mright)+T R\mleft(\frac{1}{T}\sum_{t=1}^T\cvec_t x_t\mright)\mright].
\end{align*}
Then,
\begin{equation}\label{eq:opt_rew}
    \opt(\distr) - \rew(\distr) \le \expe\mleft[(T-\tau)\bar v + \sum_{t=1}^\tau \mleft(\xi_{1,t}(\mu_t)+\xi_{2,t}(\lambdavec_t) + R(\bar \xvec_t) \mright)-T R\mleft(\frac{1}{T}\sum_{t=1}^T\cvec_t x_t\mright) \mright].
\end{equation}

\textbf{Complementary slackness term $\xi_{1,t}$.} 
For each $t$, the gradient of $\xi_{1,t}$ is given by $\nabla_\mu\xi_{1,t}=\rho - p_t x_t$ and is bounded as follows $|\nabla_\mu\xi_{1,t}|\le\rho + \bar p$. Algorithm~\ref{alg:omd} applies online mirror descent to the sequence of functions $(\xi_{1,t},\xi_{2,t})_{t=1}^T$. Therefore, by holding the second term fixed, we have that, for any pair $(\mu,\lambdavec)\in\dom$, $\sum_{t=1}^\tau \xi_{1,t}(\mu_t)-\xi_{1,t}(\mu)\le \textsc{U}(\tau,(\mu,\lambdavec))\le \textsc{U}(T,(\mu,\lambdavec))$, where $\textsc{U}(\tau,(\mu,\lambdavec))$ is the regret guarantee for online mirror descent after $\tau$ iterations as defined in~\cref{thm:omd regret}, and the second inequality follows from the fact the regret upper bound is increasing in $t$.

\textbf{Complementary slackness term $\xi_{2,t}$.}
For each $t$, the gradient of $\xi_{2,t}$ is given by $\nabla_{\lambdavec}\xi_{2,t}=\bar\xvec_t - \cvec_t x_t$. Since $\bar\xvec_t,\cvec_tx_t\in\Delta_+^m$, the gradient is bounded by $\norm{\nabla_{\lambdavec}\xi_{2,t}}_{p,\ast}\le 2$. Therefore, analogously to the previous case, for any $(\mu,\lambdavec)\in\dom$, we have $\sum_{t=1}^\tau \xi_{2,t}(\lambdavec_t)-\xi_{2,t}(\lambdavec)\le \textsc{U}(T,(\mu,\lambdavec))$. 
Now, we focus on a particular choice of $\lambdavec$. Let
\begin{equation}\label{eq:lambda_hat}
    \hat\lambdavec \in \argmax_{\lambdavec\in\mathbb{R}^m}\mleft\{ R^\ast(-\lambdavec) - \langle\lambdavec, \frac{1}{T} \sum_{t=1}^T \cvec_t x_t\rangle\mright\}.
\end{equation}
By definition of $R^\ast$ we have that, for any $t$,
\[
R^\ast\mleft(-\lambdavec\mright)\defeq \max_{\bar\xvec\in\Delta_+^m} \mleft\{ R(\bar\xvec) + \langle \lambdavec, \bar\xvec\rangle \mright\}\ge R(\bar \xvec_t) + \langle  \lambdavec ,\bar\xvec_t\rangle. 
\]
Then, by summing over $t=1,\ldots,T$ we obtain 
\begin{equation}\label{eq:conj_lambda_hat1}
R^\ast\mleft(-\lambdavec\mright)\ge\frac{1}{T}\sum_{t=1}^T \mleft(R(\bar \xvec_t) + \langle  \lambdavec ,\bar\xvec_t\rangle\mright). 
\end{equation}
Now, if $R$ is concave and closed we have that $R=R^{\ast\ast}$ (see, e.g.,~\citet[Theorem 4.2.1]{borwein2010convex}). Then, by using this fact and~\cref{eq:lambda_hat},
\begin{align*}
    R\mleft(\frac{1}{T}\sum_{t=1}^T\cvec_tx_t\mright) & = R^{\ast\ast}\mleft(\frac{1}{T}\sum_{t=1}^T\cvec_tx_t\mright) = R^\ast\mleft(-\hat\lambdavec\mright)-\langle\hat\lambdavec, \frac{1}{T}\sum_{t=1}^T\cvec_tx_t\rangle.
\end{align*}
Then,
\begin{equation}\label{eq:conj_lambda_hat2}
R^\ast\mleft(-\hat\lambdavec\mright)=R\mleft(\frac{1}{T}\sum_{t=1}^T\cvec_tx_t\mright)+\langle\hat\lambdavec, \frac{1}{T}\sum_{t=1}^T\cvec_tx_t\rangle.
\end{equation}
By~\cref{eq:conj_lambda_hat1} and~\cref{eq:conj_lambda_hat2} we obtain
\[
R\mleft(\frac{1}{T}\sum_{t=1}^T\cvec_tx_t\mright)+\langle\hat\lambdavec, \frac{1}{T}\sum_{t=1}^T\cvec_tx_t\rangle\ge \frac{1}{T}\sum_{t=1}^T \mleft(R(\bar \xvec_t) + \langle \hat\lambdavec ,\bar\xvec_t\rangle\mright).
\]
The above inequality can be rewritten as
\begin{align}
\sum_{t=1}^T\xi_{2,t}\mleft(\hat\lambdavec\mright) & = \sum_{t=1}^T\langle\hat\lambdavec,\bar\xvec_t-\cvec_tx_t\rangle \nonumber\\
& \le T R\mleft(\frac{1}{T}\sum_{t=1}^T\cvec_tx_t\mright) - \sum_{t=1}^TR(\bar\xvec_t).
\label{eq:xi_lambda_hat}
\end{align}
Finally, by setting $\hat\lambdavec$ as specified in~\cref{eq:lambda_hat}, and for any $\mu\in\mathbb{R}_{\ge0}$, we can bound the complementary slackness term as follows
\begin{align*}
    \sum_{t=1}^\tau\xi_{2,t}(\lambdavec_t)& \le \sum_{t=1}^\tau\xi_{2,t}(\hat\lambdavec)+\textsc{U}(T,(\mu,\hat\lambdavec))\\
    & = \sum_{t=1}^T\xi_{2,t}(\hat\lambdavec)-\sum_{t=\tau+1}^T\xi_{2,t}(\hat\lambdavec)+\textsc{U}(T,(\mu,\hat\lambdavec))\\
    & \le T R\mleft(\frac{1}{T}\sum_{t=1}^T\cvec_tx_t\mright) - \sum_{t=1}^TR(\bar\xvec_t)-\sum_{t=\tau+1}^T\xi_{2,t}(\hat\lambdavec)+\textsc{U}(T,(\mu,\hat\lambdavec)),
\end{align*}
where the last inequality follows from~\cref{eq:xi_lambda_hat}.

\textbf{Final step.} By substituting in~\cref{eq:opt_rew} the above upper bounds to the complementary slackness terms we obtain that, for any $\mu\in\mathbb{R}_+$ and for $\hat\lambdavec$ as in~\cref{eq:lambda_hat},
\begin{align*}
    \opt(\distr)-\rew(\distr)&\le \expe\mleft[(T-\tau)\bar v+ \sum_{t=1}^\tau \xi_{1,t}(\mu)- \sum_{t=\tau+1}^T\mleft(R(\bar\xvec_t)+\xi_{2,t}(\hat\lambdavec)\mright) \mright]+2\textsc{U}(T,(\mu,\hat\lambdavec))\\
    &\le \expe\mleft[(T-\tau)\mleft(\bar v -\underline{r} + 2\|\hat\lambdavec\|_\infty \mright)+ \sum_{t=1}^\tau \xi_{1,t}(\mu) \mright]+2\textsc{U}(T,(\mu,\hat\lambdavec)).
\end{align*}
By definition of the stopping time $\tau$, $\sum_{t=1}^\tau p_tx_t+\bar p\ge T\rho$. By letting $C=\mleft(\bar v -\underline{r} + 2\|\hat\lambdavec\|_\infty \mright)$, and $\hat\mu=C/\rho$, which is well defined as $\hat\mu\ge0$, we obtain
\begin{align*}
\sum_{t=1}^\tau \xi_{1,t}(\hat\mu) & = \sum_{t=1}^\tau \hat\mu(\rho-p_tx_t)\\
&\le \hat\mu(\tau\rho - T\rho + \bar p)\\
& =\frac{C}{\rho}\bar p - (T-\tau)C.
\end{align*}
Then,
\[
\opt(\distr)-\rew(\distr)\le \frac{C}{\rho}\bar p + 2\textsc{U}(T,(\mu,\lambdavec)).
\]
Finally, we observe that $\|\hat\lambdavec\|_\infty\le L$ because $R$ is $L$-Lipschitz continuous with respect to $\|\cdot\|_1$ and $\hat\lambdavec\in\partial R(\bar x_t)$ by~\cref{eq:lambda_hat}. Therefore, by substituting the regret guarantees for OMD as specified in~\cref{thm:omd regret}, we have 
\[
\opt(\distr)-\rew(\distr)\le C_1 + \frac{G^2\eta}{\sigma} T +\frac{1}{\eta}C_3,
\]
where $C_1=(\bar v -\underline{r} + 2L)\bar p/\rho$, $G=\max\{\rho+\bar p,2\}$, and 
\[
C_3=2\sup\mleft\{B_\psi((\mu,\lambdavec),(\mu_1,\lambdavec_1)):(\mu,\lambdavec)\in\dom,\norm{\lambdavec}\le L\mright\}.
\] 
This concludes the proof.
\end{proof}

\section{Additional experimental results}\label{sec:exp appendix}

In order to further assess the performance of Algorithm~\ref{alg:omd} we consider the following setting. We start from the \textsc{B2} scenario described in~\cref{sec: exp eval}, where we have two possible categories per item corresponding to different mobile device types. 
Then, we artificially create an unbalance in the price distributions for the two categories, making one category slightly cheaper than the other. 
~\cref{fig:exp cont} reports the comparison between Algorithm~\ref{alg:omd} and \textsc{AP}. 
\cref{fig:exp cont} (Right) show that the total variation distance between the realized distribution of impressions $\bar\xvec$ and the target distribution $\hat\xvec$ is close to $0.5$. Since $\hat\xvec=(.5,.5)$ this means that \textsc{AP} is forcing the bidder to win \emph{only} impressions for \emph{cheap} requests. 
On the other hand, by employing Algorithm~\ref{alg:omd} we can effectively steer the realized distribution of impression towards the target without compromising the quality of the solution in terms of rewards. 

\begin{figure}
\begin{minipage}{.45\textwidth}
\includegraphics[scale=0.3]{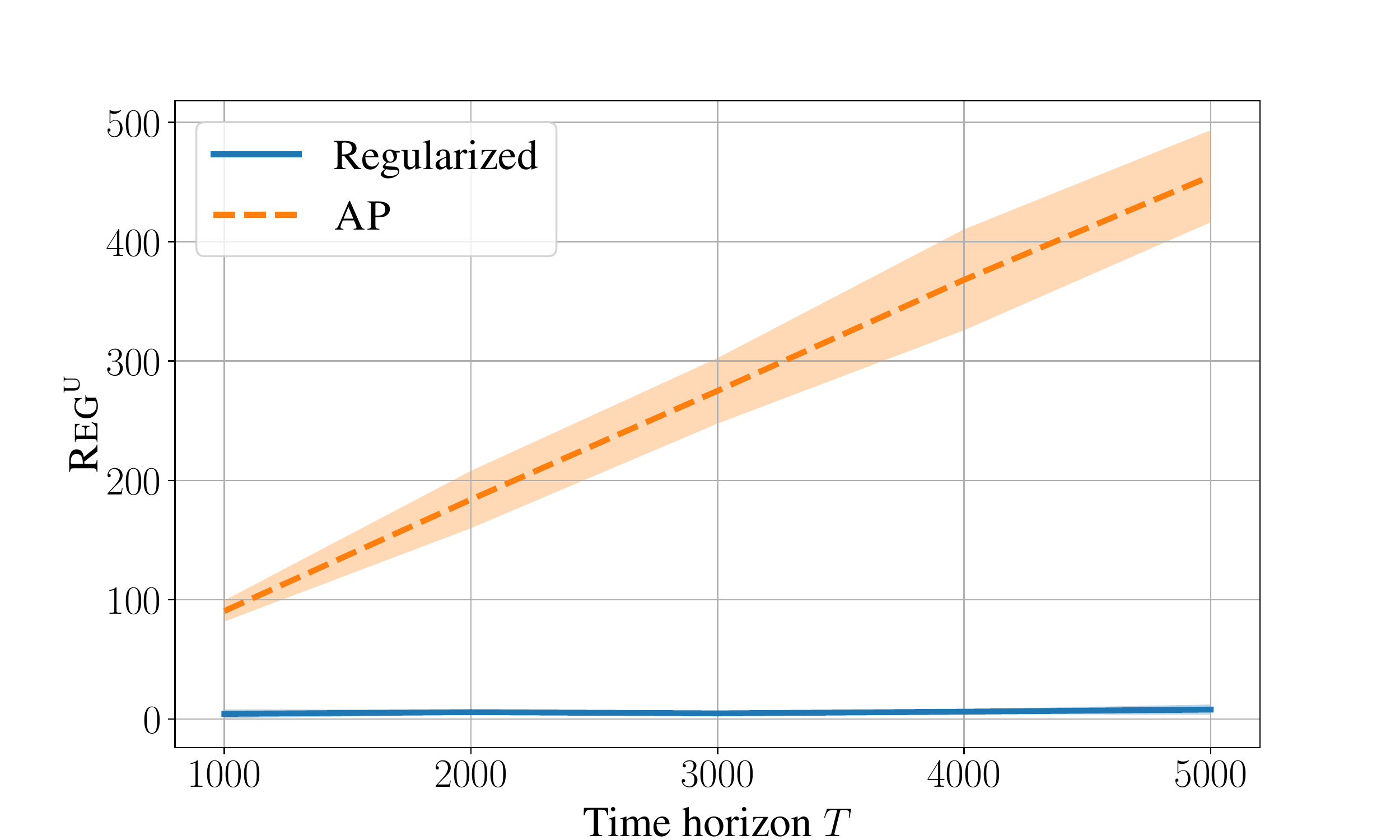}
\end{minipage}
\hspace{.5cm}
\begin{minipage}{.45\textwidth}
\includegraphics[scale=0.3]{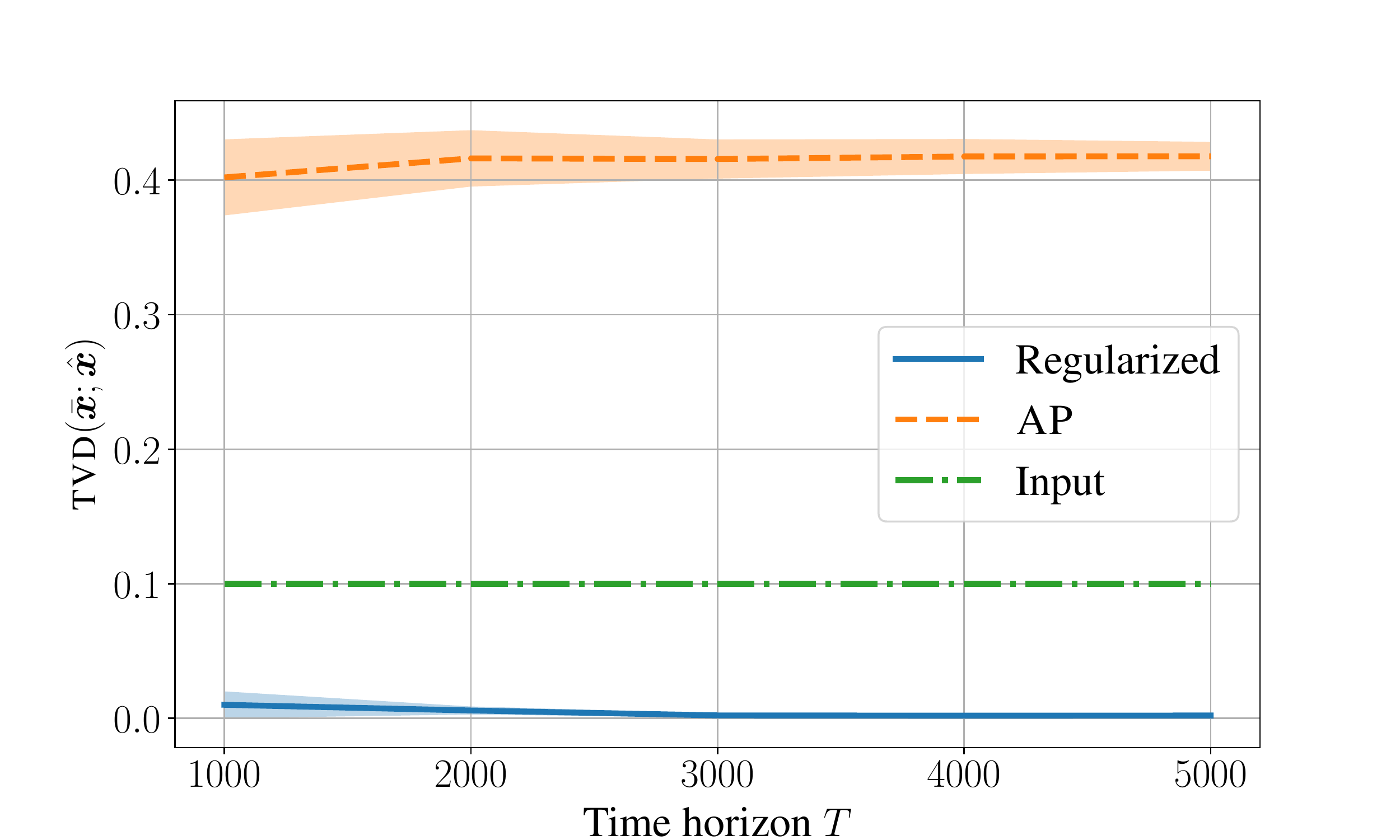}
\end{minipage}
\caption{\footnotesize \textbf{(Left)}: Empirical upper bound on the regret (i.e., $\reg^\textsc{u}$ as defined in~\cref{eq:regret ub}) for breakdown \textsc{B2} with prices unbalance. \textbf{(Right)}: Total variation distance from the target distribution $\hat\xvec$ for \textsc{B2} with prices unbalance.}
\label{fig:exp cont}
\end{figure}

\end{document}